\documentclass{svproc}
\usepackage{url}
\usepackage{amsmath}
\usepackage{pgfplots}
\pgfplotsset{compat=1.15}
\usepackage{mathrsfs}
\usetikzlibrary{arrows}
\usepackage{graphicx}

\begin{document}
\mainmatter              
\title{Quantum fractional revival governed by adjacency matrix Hamiltonian in unitary Cayley graphs}
\titlerunning{Quantum Fractional Revival governed by adjacency matrix Hamiltonian}  
%
\author{Rachana Soni \and Neelam Choudhary
\and
Navneet Pratap Singh}
\authorrunning{Rachana Soni et al.} 
%
\tocauthor{Rachana Soni, Neelam Choudhary, Navneet Pratap Singh}
\institute{Bennett University, Greater Noida, U.P.,\\
\email{E20SOE813@bennett.edu.in},
}

\maketitle              

 \begin{abstract}
 \noindent
In this article, we give characterization for existence of quantum fractional revival in unitary Cayley graph utilizing adjacency matrix Hamiltonian. Unitary Cayley graph $X=( Z_n, S)$ is a special graph as connection set $S \subseteq Z_n$ is the collection of   coprimes to $n$. Unitary Cayley graph is an integral graph and its adjacency matrix is a circulant one. We prove that quantum fractional revival in unitary Cayley graphs exists only when the number of vertices is even. Number-theoretic and spectral characterizations are given for unitary Cayley graph admitting quantum fractional revival.  Quantum fractional revival is analogous to quantum entanglement. It is one of qubit state transfer phenomena useful in  communication of quantum information.  
 \\
 
 \end{abstract}

\section{Introduction}
\hspace{3mm}{The theory of quantum walks, the quantum version of classical walks in graph network, is interesting and exciting because it can be utilized to impact quantum computing algorithms. The continuous time quantum walk (CTQW) stands out for its simplicity and its relevance in creating faster quantum algorithms. In this paper, we present the basic ideas of CTQWs. CTQW is basically emphasizing that their evolution is controlled by the solution of Schrödinger equation instead of a sequence of probabilistic steps. CTQWs have been useful in uncovering insights into quantum dynamics, providing the investigation of quantum transport phenomena, and simulating complex quantum systems. In quantum walk, there are few quantum state transfer phenomenons who tell us about probability of teleportation of qubit states- quantum fractional revival(QFR),  perfect state transfer(PST),  prefect edge state transfer(PEST), pretty good fractional revival(PGFR), pretty good state transfer(PGST), are few mathematical framework which are researched extensively recently.}

{The transition matrix of quantum walks tells us about evolution of walker's position on path over time. This operator gives us knowledge about how quantum information flows in a graph which is representing network connection.}

{In this paper, we are studying quantum fractional revival on unitary Cayley graphs. The base theory of QFR with respect to adjacency hamiltonian is explored by \cite{GV}, \cite{BPA} and \cite{Achan1}. QFR governed by Laplacian matrix hamiltonian is studied thoroughly by \cite{LA}, they proved the characterization on cycle and path graph by taking Laplacian matrix as hamiltonian. 

{Let's go through the introductory knowledge of some basic features and functions of the quantum walk. }

\subsection{Evolution operator for the quantum walker}

\hspace{3mm} {The Hamiltonian \( H \) affects the evolution state of a quantum walker being transported in the network, for which the adjacency matrix is considered as hamiltonian often. The evolution operator over the time $t$, \( U(t) \), defined as:
\[
U(t) = e^{-iHt}
\]
here \( i \) is the complex imaginary unit,
    Hamiltonian of the quantum network system is \( H \), here setting to the \( A \), the adjacency matrix,
     \( t \) representing time.}

{The operator \( U(t) \) is a unitary, fixing the total probability of the quantum state remains 1 over time. The unitary nature of the evolution operator \( U(t) \) also implies that the evolution can be reversed, a crucial property of quantum processes.}

{From the concept  of quantum mechanics, \( U(t) \) acts on a qubit state vector \( \psi(0) \) representing the initial state of the system and describes the final quantum state. In the previous extensive research of events of quantum state transfer, there are significant state transfer phenomena that show the ability of the quantum system to transport quantum states between nodes with high fidelity and minimum loss of information.} 

\subsection{Periodicity}

Periodicity in continuous time quantum walks describes the phenomenon where a qubit state  at a vertex \( x \) returns to itself exactly after some time \( t \). That is, the graph admits periodicity at vertex \( x \) if there exists a time \( t \) such that:
\[
U(t) e_x = \gamma e_x
\]
for  phase factor \( \gamma \) with modulus unity.  \( U(t) = e^{-iAt} \) as defined earlier. Periodicity indicates the quantum walk is  symmetric at \( u \).

\subsection{Perfect State Transfer }

The transportation phenomenon PST, exists in a graph when a quantum state initially at one vertex of a graph is transferred entirely to another vertex after a certain time. Mathematically, a graph \( G \) with an evolution operator \( U(t) \) exhibits PST, if for  \( u \) and \( v \) we have:
\[
U(t) e_u = \gamma e_v
\]
for some time \( t \) and phase \( \gamma  \),  \( e_v \) and \( e_u \) are the standard basis state vectors at  \( u \) and \( v \) vertices respectively. Here, \( U(t) = e^{-iAt} \) and \( A \) is the Hamiltonian of the system. PST is special case of QFR which is widely studied in the paper \cite{Bose},\cite{G1},\cite{Cthesis},\cite{LPST},\cite{Zhou},\cite{BM0},\cite{BM},\cite{CM1},\cite{CM2},\cite{ACRJ}.

\subsection{Quantum fractional revival}

Quantum Fractional Revival happens when an initial quantum state at a vertex partially rebuilds itself with quantum state of other vertex  after some time $t$. For a graph \( X \) with evolution operator \( U(t) \), quantum fractional revival is observed between\( u \) and \( v \) vertices at $t$ if:
\[
U(t) e_u = \alpha e_u + \beta e_v 
\]
where we have condition \( |\alpha|^2 + |\beta|^2 = 1 \), with the complex numbers \( \alpha \) and \( \beta \). This implies that the qubit state at  \( u \) walks over \( v \) while keeping significant amplitude at the original vertex \( u \).

These quantum mechanics phenomena highlight the rich dynamics and potential applications of quantum walks in quantum communication and computation, particularly in tasks related to quantum state transfer and quantum entanglement generation.
. }

{Cayley graphs provide an effective algebraic structure and visual explanation of groups. These graphs are generated by visulizing group elements as vertices and considering a subset of the group's elements known as generators to create edges between the vertices. The generators resulting the group operation from a vertex to another that is, edges of a Cayley graph. Cayley graphs are a vital tool in group theory and its applications in many domains, including algebraic topology, applications of group theory and network analysis. This structure offers a deep insight into the symmetry and structure of the graph generated by groups. The degree of nodes in Cayley graph is happened to be $|S|$. }
\section{Preliminaries}
\subsection{Unitary Cayley graph}

{Thw Unitary Cayley graphs $X=(Z_n,S)$ are the kind of Cayley graphs where the edge generators set $S$ is collection of all those elements of $Z_n$ which are coprime to the number $n$, that is, $S=\{u:gcd(u,n)=1$ or in other words collection of units in the group $Z_n.$ The additive cyclic group $Z_n$ will be working as set of vertices and connection set $S$ decides the connection between vertices and will form set of edges.} \\

\noindent Some properties:

\begin{enumerate}
 
    \item Unitary Cayley graph is always a connected graph.

    \item if number of vertices of graphs is a prime then it becomes a complete graph.

    \item if number of vertices is equal to some power of prime  then it becomes complete $p-$partite graph.

    \item It is a -regular graph with $|\phi(n)|$ degree, $\phi(n)$ is well-known Euler-phi function.

    \item Adjacency matrix of it is a circulant, Hermitian and symmetric. Eigenvalues of it are always integers. These properties make it a remarkable graph. One can see results on integral and circulant graphs in \cite{A}, \cite{eb}, \cite{BM}, \cite{NS}, \cite{ACRJ}, \cite{UCG}.
  
\end{enumerate}

\begin{figure}
\vspace{-2cm}
\definecolor{yqqqqq}{rgb}{0.5019607843137255,0.,0.}
\begin{tikzpicture}[line cap=round,line join=round,x=0.8149127445795874cm,y=1.0cm]
\clip(-2.58,-2.44) rectangle (15,5);
\draw [line width=1.pt] (-1.82,-0.02)-- (0.,0.);
\draw [line width=1.pt] (1.34,1.)-- (0.36,-1.06);
\draw [line width=1.pt] (0.36,-1.06)-- (2.3,-1.);
\draw [line width=1.pt] (1.34,1.)-- (2.3,-1.);
\draw [line width=1.pt] (3.,1.)-- (3.,-1.);
\draw [line width=1.pt] (3.,-1.)-- (5.,-1.);
\draw [line width=1.pt] (5.,1.)-- (5.,-1.);
\draw [line width=1.pt] (3.,1.)-- (5.,1.);
\draw [line width=1.pt] (8.,2.)-- (6.16,0.96);
\draw [line width=1.pt] (6.16,0.96)-- (7.,-1.);
\draw [line width=1.pt] (7.,-1.)-- (9.,-1.);
\draw [line width=1.pt] (9.,-1.)-- (10.,1.);
\draw [line width=1.pt] (8.,2.)-- (10.,1.);
\draw [line width=1.pt] (8.,2.)-- (7.,-1.);
\draw [line width=1.pt] (8.,2.)-- (9.,-1.);
\draw [line width=1.pt] (6.16,0.96)-- (10.,1.);
\draw [line width=1.pt] (6.16,0.96)-- (9.,-1.);
\draw [line width=1.pt] (7.,-1.)-- (10.,1.);
\draw [line width=1.pt] (12.,2.)-- (11.,1.);
\draw [line width=1.pt] (11.,1.)-- (11.,-1.);
\draw [line width=1.pt] (11.,-1.)-- (12.,-2.);
\draw [line width=1.pt] (12.,-2.)-- (13.,-1.);
\draw [line width=1.pt] (13.,-1.)-- (13.,1.);
\draw [line width=1.pt] (12.,2.)-- (13.,1.);
\begin{scriptsize}
\shade [ball color= red!70!black] (-1.82,-0.02) circle (2.5pt);
\shade [ball color= red!70!black] (0.,0.) circle (2.5pt);
\shade [ball color= red!70!black] (1.34,1.) circle (2.5pt);
\shade [ball color= red!70!black] (0.36,-1.06) circle (2.5pt);
\shade [ball color= red!70!black] (2.3,-1.) circle (2.5pt);
\shade [ball color= red!70!black] (3.,1.) circle (2.5pt);
\shade [ball color= red!70!black] (5.,1.) circle (2.5pt);
\shade [ball color= red!70!black] (3.,-1.) circle (2.5pt);
\shade [ball color= red!70!black] (5.,-1.) circle (2.5pt);
\shade [ball color= red!70!black] (8.,2.) circle (2.5pt);
\shade [ball color= red!70!black] (6.16,0.96) circle (2.5pt);
\shade [ball color= red!70!black] (7.,-1.) circle (2.5pt);
\shade [ball color= red!70!black] (9.,-1.) circle (2.5pt);
\shade [ball color= red!70!black] (10.,1.) circle (2.5pt);
\shade [ball color= red!70!black] (12.,2.) circle (2.5pt);
\shade [ball color= red!70!black] (11.,1.) circle (2.5pt);
\shade [ball color= red!70!black] (11.,-1.) circle (2.5pt);
\shade [ball color= red!70!black] (12.,-2.) circle (2.5pt);
\shade [ball color= red!70!black] (13.,-1.) circle (2.5pt);
\shade [ball color= red!70!black] (13.,1.) circle (2.5pt);
\end{scriptsize}
\end{tikzpicture}

\caption{Unitary Cayley graph examples for $2,3,4,5,6$ vertices respectively.}
\end{figure}
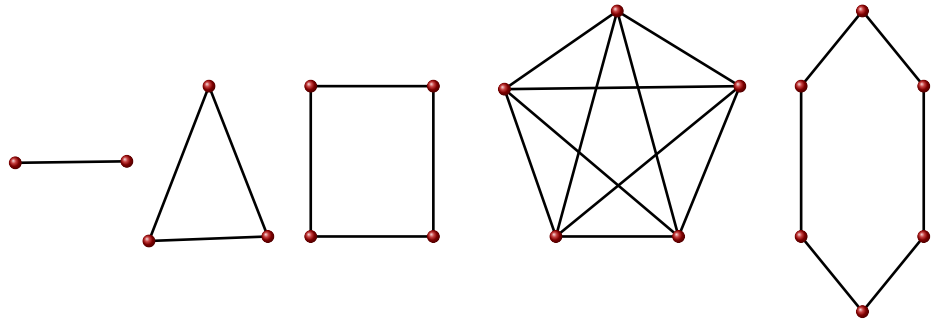

Adjacency matrix of unitary cayley graphs is defines as: $$[a]_{uv} = 1~\text{if}~gcd(u-v,n) =1, \text{otherwise}~ [a]_{uv}=0.$$.

\subsection{Eigenvalues of unitary Cayley graph}
The information about eigenvalues and eigenvectors of unitary Cayley graph is given in \cite{Liu}, 
$$ \lambda_{d} = \sum_{j\in S}\omega^{d j} = r(d,n) $$ 
Here $r(d,n)$ is Ramanujan sum function \cite{UCG}, $$r(d,n) =\mu (t_{d}) \dfrac{\phi(n)}{\phi(t_{d})}$$ 
where $t_{d}= \dfrac{n}{gcd(d,n)}.$ \\

The eigenvectors are, 
$$v_{d} = [1 ~ \omega_n^{d}~ \ldots ~ \omega_n^{(n-1)d}]$$

\subsection{Evolution operator of unitary Cayley graph} 
$$ U(t) = \dfrac{1}{n} \sum_{d=0}^{(n-1)} e^{i \lambda_{d}t}E_{d}$$

where $\lambda_{d}$ are distict eigenvalues of adjacency matrix and $E_{d}$ is projection matrix that is $E_{d} = v \times v^*$

\subsection{Properties of projection matrices $E_{d}$}

\begin{enumerate}

    \item Projection matrices are idempotent matrix that is $E_{d}^2 = E_{d}$.
    
    \item Projection matrix equals to its conjugate transpose hence it is a Hermitian matrix.
    \item It has eigenvalues $0$ and $1$. and $0$ eigenvalue has multiplicity $(n-1).$

    \item Trace of the projection matrix will always be 1.
    \item It is orthogonal to other projection matrices of the graph, that is, $$E_{d} E_{e}=0, d \neq e$$
    \item Sum of projection matrices is identity matrix $$\sum_{d}^{(n-1)} E_{d} = I $$
\item It is 1-rank matrix because it projects on to the one dimensional subspace spanned by eigenvector $v.$
\item Two vertices $u,v$ are called strongly cospectral if $$E_d e_u= E_d e_v `~ \text{for every}~ d .$$

\end{enumerate}

 Quantum fractional revival is said to be  \textit{proper} if $\beta \neq 0$, It is called \textit{balanced fractional revival} if $|\alpha|$ equals to $|\beta|$.

{If QFR occurs in the graph then, $$e^*_u U(t) e_u = e^*_u\alpha e_u + e^*_u\beta e_v$$ since $e_u$ and $e_v$ are orthonormal vectors,
$$\alpha = U(t)_{u,u}= \dfrac{1}{n}  \sum_{d=0}^{n-1}e^{i\lambda_d t}$$
and $$ \beta = U(t)_{v,u}=\dfrac{1}{n} \sum_{d =0}^{n-1}e^{i \lambda_{d}t}\omega_n^{(v-u)d}$$}

\begin{remark}  \label{remark:1.2} \cite{AChan} If first two rows and columns of evolution operator of quantum walk $U(t)$ and adjacency matrix $A$ are labelled  $u$ and $v$. Now if QFR occurs  between $u$ and $v$ at  $t$ time, then,

\[
U(t) = \begin{bmatrix}
P & 0 \\
0 & Q
\end{bmatrix},
\]

where

\[
P = \begin{bmatrix}
\alpha & \beta \\
\beta & \gamma
\end{bmatrix},
\]

and $\gamma = -\frac{\overline{\alpha}}{\overline{\beta}} \beta$ if $\beta \neq 0$. 

\end{remark} 

\begin{definition} \cite{BM}  M\"obius function \( \mu(m) \), for $m \geq 0 $, is the summation of the {primitive \(m\)th roots of unity}. Its value lies in set \(\{-1, 0, 1\}\) which depends upon the {prime factorization} of \( m \):

 \begin{align*}
     \mu(m) &= 
\begin{cases} 
1, & m~ \text{ has prime factors in even numbers and is squarefree.}\\ 
-1, & m~ \text{has  prime factors in odd numbers and is squarefree.} \\
0, & m ~ \text{has a squared prime factor.}
\end{cases} \\
 \end{align*}

\end{definition}

\section{Main Results}\label{sec-3}

\begin{lemma}
    If $u$ and $v$ are the vertices of unitary Cayley graph $X=(Z_n,S)$  then $(E_d)_{uu} = (E_d)_{vv} =1/n$, moreover the $uu^{th}$ and $vv^{th}$ entries of evolution operator of continuous time quantum walk  $(U_d)_{uu} = (U_d)_{vv}$ for every $u,v \in X$.
\end{lemma}
\begin{proof}
    The eigenvector of unitary Cayley graph $v_r = [1 \omega^r \omega^{2r} \ldots \omega^{(n-1)r}]^T $. The spectral idempotent matrix $$(E_{d})_{uv} =  \nu_d \times \nu_d^* = \dfrac{1}{\sqrt{n}} \times \dfrac{1}{\sqrt{n}} [\omega_n^{(v-u)d}]= \dfrac{1}{n} [\omega_n^{(v-u)d}]$$ It is clear that for the vertices $u$ and $v$, $$ (E_d)_{uu} = (E_d)_{vv} = 1/n$$ 
    The evolution operator for unitary Cayley graph is $$ U(t) =  \sum_{d=0}^{(n-1)} e^{i \lambda_{d}t}E_{d}$$, for the vertices $u$ and $v$, the entry $$U(t)_{uu} = U(t)_{vv}=  \sum_{d=0}^{(n-1)} e^{i \lambda_{d}t}$$
\end{proof}
\begin{lemma}
    If unitary Cayley graph $X=(Z_n,S)$ admits QFR between the nodes $u$ and $v$ then $u$ and $v$ are strongly cospectral.
\end{lemma}
\begin{proof}
    Let $u$ and $v$ be the vertices in the unitary Cayley graph $X$ where QFR occurs at the time $t$, Now if $U(t)$ be the evolution operator of CTQW then, 

     $$U(t) e_u = \alpha e_u + \beta e_v$$\\
    $$ \sum_{d=0}^{(n-1)} e^{it \lambda_{d} }E_{d}e_u = \alpha e_u + \beta e_v$$
 multiplying $E_{d}$ both side  knowing the fact that $E_{\ell}$ is orthogonal projection,  
  $$  e^{i t \lambda_{d} }E_{d}e_u = \alpha E_{d} e_u + \beta E_{d} e_v$$
  $$ ( e^{i t \lambda_{d} }- \alpha )E_{d} e_u = \beta E_{d} e_v$$
  implies $E_{d} e_u$ is scalar multiple of $E_{d} e_v$ but as we proved in previous lemma, for unitary Cayley graph, $(E_r)_{uu} = (E_r)_{vv}$, we can conclude that $$E_{d} e_u = \pm E_{d} e_v$$ That is, $u$ and $v$ are strongly cospectral vertices. 

\end{proof}
\begin{lemma} \label{thm-1}
    
 Let $X=(Z_n,S)$ be a unitary Cayley graph and $\alpha, \beta$ are complex numbers then graph $X$ admits quantum fractional revival between $u$ and $v$ at  $t >0$ time if and only if the unitary evolution operator is of the format \begin{equation}U(t) = \begin{bmatrix}
    \begin{bmatrix}
        \alpha & \beta \\ \beta & \overline{\alpha}
    \end{bmatrix}& \textbf{0} \\ \textbf{0} & Q
\end{bmatrix} \label{eq:Ut}\end{equation}   \end{lemma}

\begin{proof}
First we prove necessary condition, suppose there is QFR between vertices $u$ and $v$, by remark \ref{remark:1.2}  evolution operator will be a block diagonal matrix. \\

Now if  the first two rows and columns of $U(t)$ are labelled as $u$ and $v$ and $U(t)$. As we know diagonal entries of $U(t)$ for the existence of QFR, $|U(t)_{(u,u)}| = \alpha$, $|U(t)_{(v,u)}| = \beta $, hence, for unitary Cayley graph, the matrix $$P= \begin{bmatrix}
     \alpha & \beta \\ \beta & \overline{\alpha}
 \end{bmatrix} $$. \\

 Conversely, let unitary evolution operator is of the format \ref{eq:Ut}), $(u,u)^{th}$ and $(v,u)^{th}$ entries of $U((t)$ are, $\alpha$ and $\beta$ respectively. Now we prove $|\alpha|^2 + |\beta|^2 =1 $
 for existence of quantum fractional revival, since evolution operator is unitary, hence, \begin{align*}
 U(t)U^*(t)&=I \\
 \alpha \overline{\alpha} + \beta \overline{\beta} &=1 ~ \text{and} ~ \alpha \overline{\beta} + \beta \overline{\alpha} = 0 \\
 |\alpha|^2 + |\beta|^2 &= 1. ~ 
 \end{align*}
\end{proof}

\begin{theorem} \label{thm-2}
    If $n$ is odd, there is no quantum fractional revival in unitary Cayley graph.
\end{theorem}
\begin{proof}
    Quantum fractional revival exists between anti-podal vertices in unitary Cayley graphs and if $n$ is odd, there will be even numbers of anti-podal vertices for every vertex, by the consequence of proposition 6.4 of \cite{Achan1}, $n$ can not be odd. 
\end{proof}

\begin{corollary} \label{thm-3}
    A unitary cayley graph admits quantum fractional revival if and only if n is an even number.
\end{corollary}
\begin{theorem} \label{thm-4}
    If a unitary Cayley graph $X =(Z_n,S)$ exhibits QFR between nodes $u$ and $v$ at $t$ time, where $n$ is a squarefree integer, and eigenvalue supports $\varsigma^+_{u,v} = \{ \lambda_1, \lambda_2 \ldots \} $ $\varsigma^-_{u,v} = \{ \widetilde{\lambda_1}, \widetilde{\lambda_2}, \ldots \}$ then $\lambda_1- \widetilde{\lambda_1} = 0 ~ mod~2. $ \\
    Here, $\varsigma^+_{u,v} = \{\lambda_d: E_de_u = E_de_v\}$ and $\varsigma^-_{u,v} = \{\lambda_d: E_de_u = -E_de_v\}$
\end{theorem}
\begin{proof}
    By the corollary 14 of \cite{UCG}, $\{ 1, -1\}$ both are eigenvalues of the unitary Cayley
graph with  $\phi(n)$ multiplicity, if and only if $n$ is squarefree
and even number. Now, since QFR exists only if $n$ is an even number in unitary cayley graph and $1$ and $-1$ will be in the different eigenvalue support due to opposite signs, hence the results.
\end{proof}

\begin{example}
 Quantum fractional revival occurs in unitary Cayley graph for $n= 2,4$ and $6$ between the antipodal vertices at the time $t= \pi/2, \pi/2$ and $2\pi/3$ respectively.
\end{example}

\section{Conclusion} We studied quantum fractional revival in unitary Cayley graph, which is an important phenomenon of quantum state transfer of continuous time quantum walks. Unitary Cayley graph is a specific case where set of edge-generators $|S|$  is collection of units with respect to $n.$ Quantum fractional revival which is quantum analogue of quantum entanglement, is explored. We proved that quantum fractional revival exists in unitary Cayley graphs only if the vertices count $n$ is even. We also computed the form of unitary evolution matrix of the quantum walk in unitary Cayley graph and spectral characterization if QFR exists. 
\section*{Acknowledgment}

Authors are grateful to Bennett University for providing financial support.


\end{document}